\documentclass[5p]{elsarticle}
\usepackage{hyperref,breakurl,color}
\usepackage{bm}
\usepackage{stmaryrd}
\usepackage{amssymb}
\usepackage{amsmath}
\usepackage{amsthm}
\usepackage{enumerate}
\usepackage[normalem]{ulem}
\makeatletter
\def\ps@pprintTitle{%
     \let\@oddhead\@empty
     \let\@evenhead\@empty
     \let\@oddfoot\@empty
     \let\@evenfoot\@oddfoot}
\makeatother
\newfont{\bbb}{bbm10}                    
\newfont{\bbbs}{bbm10 scaled 700}        

\theoremstyle{plain}

\newtheorem{theorem}{Theorem}
\newtheorem{lemma}{Lemma}
\newtheorem{corollary}{Corollary}

\theoremstyle{remark}
\newtheorem{remark}{Remark}
\newtheorem{observation}{Observation}

\theoremstyle{definition}
\newtheorem{definition}{Definition}
\newcommand\weg[1]{} 
\newcommand{\plat}[1]{\raisebox{0pt}[0pt][0pt]{#1}}     

\newcommand{\subs}[2]{\{\mathord{\raisebox{2pt}[0pt]{$#1$}\!/\!#2}\}} 
\newcommand{\trans}[1]{\mbox{\bbb [}#1\mbox{\bbb ]}}
\newcommand{\subtrans}{\mbox{\bbbs [\ ]}}
\newcommand{\op}{\mathtt{op}}
\newcommand{\T}{\mathcal{P}}           
\newcommand{\N}{\mathcal{N}}           
\newcommand{\bn}{\mathrm{bn}}
\newcommand{\fn}{\mathrm{fn}}
\newcommand{\n}{\mathrm{n}}
\renewcommand{\nu}{}
\newcommand{\nil}{\mathbf{0}}
\newcommand{\sbarb}[2]{#1{\downarrow}}
\newcommand{\equred}{\mathrel{\equiv}}
\newcommand{\eqused}{\mathrel{\equiv_{\mathrm{S}}}}
\newcommand{\C}{U}
\newcommand{\D}{V}
\newcommand{\E}{T}
\newcommand{\F}{W}
\newcommand{\pT}{T}

\newcommand{\pims}{\pi}
\newcommand{\pima}{{\rm a\pi}}

\title{On the Validity of Encodings of the Synchronous in the Asynchronous $\pi$-calculus}
\author[csiro,nsw]{Rob J. van Glabbeek}
\ead{rvg@cs.stanford.edu}

\address[csiro]{Data61, CSIRO, Sydney, Australia}
\address[nsw]{School of Computer Science and Engineering, University of New South Wales, Sydney, Australia}

\begin{document}

\begin{abstract}
Process calculi may be compared in their expressive power by means of encodings between them.
A widely accepted definition of what constitutes a valid encoding for (dis)proving relative
expressiveness results between process calculi was proposed by Gorla.
Prior to this work, diverse encodability and separation results were generally obtained using
distinct, and often incompatible, quality criteria on encodings.

Textbook examples of valid encoding are the encodings proposed by Boudol and by Honda \& Tokoro of
the synchronous choice-free $\pi$-calculus into its asynchronous fragment, illustrating that the
latter is no less expressive than the former.
Here I formally establish that these encodings indeed satisfy Gorla's criteria.
\end{abstract}


\maketitle

\section{Introduction}

\noindent
Since the late 1970s, a large number of process calculi have been
proposed, such as CCS \cite{Mi80}, CSP \cite{BHR84}, ACP \cite{BK86acopy},
SCCS~\cite{Mi83copy}, {\sc Meije} \cite{AB84}, LOTOS \cite{BB87}, the
$\pi$-calculus \cite{MPW92}, mobile ambients
\cite{CG00} and mCRL2 \cite{GM14}. To cater to specific applications,
moreover many variants of these calculi were created, including versions
incorporating notions of time, and probabilistic choice.

To order these calculi w.r.t.\ expressiveness, encodings between them
have been studied \cite{dS85,San93,Va93copy,vG94a,
Boreale98,Parrow00,
NestmannP00,Palamidessi03,Nestmann00,CG00,QW00,BusiGZ09,CarboneM03,
BPV04,
Palamidessi05,PalamidessiSVV06,PV06,CCP07,VPP07,CCAV08,
HMP08,PV08,VBG09,
EPTCS190.5,PN16}. 
Process calculus $\mathcal{L}_1$ is said to be
\emph{at least as expressive as} process calculus $\mathcal{L}_2$ iff
there exists a valid encoding from $\mathcal{L}_1$ into $\mathcal{L}_2$.
However, in proving that one languages is---or is not---at least as
expressive as another, different authors have used different, and
often incomparable, criteria.

Gorla \cite{Gorla10a} collected some essential features of the above
approaches and integrated them in a proposal for a valid encoding that
justifies many encodings and separation results from the literature.
Since then, many authors have used Gorla's framework as a basis for
establishing new valid encodings and separation results
\cite{Gorla10b,LPSS10,PSN11,PN12,PNG13,GW14,EPTCS160.4,EPTCS189.9,GWL16}.

Often quoted token examples of valid encodings
\cite{Palamidessi03,Nestmann00,CC01,CarboneM03,CCP07} are the encodings
proposed by Boudol \cite{Bo92} and by Honda \& Tokoro \cite{HT91} of
the synchronous choice-free $\pi$-calculus into its asynchronous
fragment, illustrating that the latter is as expressive as
the former. Gorla mentions these encodings among his first three
examples of encodings that satisfy his criteria for validity
\cite{Gorla10a}, thereby giving evidence in support of his combination
of criteria, more than in support of these encodings.
Nevertheless, I have not found a proof in the literature that these
encodings satisfy Gorla's notion of validity, nor is the matter
trivial. The goal of this paper is fill this gap and
formally establish that the encodings of \cite{Bo92,HT91} indeed
are valid \`a la Gorla.

Section~\ref{sec:validty} recalls Gorla's proposal for validity of an
encoding; for their motivation see \cite{Gorla10a}.
Section~\ref{sec:pi2api} presents the encodings of \cite{Bo92} and
\cite{HT91}, again suppressing motivation, and
Sections~\ref{sec:valid}--\ref{sec:validHT} establish their validity.
Section~\ref{sec:conclusion} reflects back on Gorla's criteria in the
light of the present application, and compares with the notion of a
valid encoding from \cite{vG12}.

\section{Valid encodings}\label{sec:validty}

\noindent
In \cite{Gorla10a} a \emph{process calculus} is given as a triple
$\mathcal{L}\mathbin=(\mathcal{P},\longmapsto,\asymp)$, where
\begin{itemize}
\item $\mathcal{P}$ is the set of language terms (called
  \emph{processes}), built up from $k$-ary composition operators $\op$.
\item $\longmapsto$ is a binary \emph{reduction} relation between
  processes.
\item $\asymp$ is a semantic equivalence on processes.
\end{itemize}
The operators themselves may be constructed from a set $\N$
of names. In the $\pi$-calculus, for instance, there is a unary
operator $\bar x y.\_$ for each pair of names
$x,y\mathbin\in\N$.
This way names occur in processes; the occurrences of names in
processes are distinguished in \emph{free} and \emph{bound} ones;
$\fn(\vec P)$ denotes the set of names occurring free in the $k$-tuple
of processes $\vec P=(P_1,\dots,P_k)\mathbin\in\mathcal{P}^k$.
A \emph{renaming} is a function $\sigma:\N\rightarrow\N$;
it extends componentwise to $k$-tuples of names.
If $P\mathbin\in\mathcal{P}$ and $\sigma$ is a renaming, then $P\sigma$
denotes the term $P$ in which each free occurrence of a name $x$ is
replaced by $\sigma(x)$, while renaming bound names to avoid name capture.

A $k$-ary $\mathcal{L}$-\emph{context} $C[\__1;\dots;\__k]$ is a term
build by the composition operators of $\mathcal{L}$ from \emph{holes}
$\__1,\dots,\__k$; each of these holes must occur exactly once in the context.
If $C[\__1;\dots;\__k]$ is a $k$-ary $\mathcal{L}$-\emph{context} and
$P_1,\dots,P_k \in \mathcal{P}$ then $C[P_1;\dots;P_k]$ denotes the
result of substituting $P_i$ for $\__i$ for each $i\mathbin=1,\dots,k$,
while renaming bound names to avoid capture.

Let $\Longmapsto$ denote the reflexive-transitive closure of $\longmapsto$.
One writes $P\longmapsto^\omega$ if $P$ \emph{diverges}, that is, if
there are $P_i$ for $i\in\mbox{\bbb N}$ such that $P\mathbin= P_0$ and
$P_i\longmapsto P_{i+1}$ for all $i\mathbin\in\mbox{\bbb N}$.
Finally, write $P\longmapsto$ if $P\longmapsto Q$ for some term $Q$.

For the purpose of comparing the expressiveness of languages,
a constant $\surd$ is added to each of them~\cite{Gorla10a}.
A term $P$ in the upgraded language is said to \emph{report success},
written $P{\downarrow}$, if it has an \emph{top-level unguarded} occurrence
of $\surd$.\footnote{Gorla defines the latter concept only for languages
  that are equipped with a notion of \emph{structural congruence} $\equiv$ as well as a parallel
  composition $|$. In that case
  $P$ has a top-level unguarded occurrence of $\surd$ iff $P\equiv Q|\surd$, for some $Q$~\cite{Gorla10a}.
  Specialised to the $\pi$-calculus, a \emph{(top-level) unguarded} occurrence is one that not lays
  strictly within a   subterm $\alpha.Q$, where $\alpha$ is $\tau$, $\bar xy$ or $x(z)$.
  For De Simone languages \cite{dS85}, even when not equipped with $\equiv$ and $|$, a suitable
  notion of an unguarded occurrence is defined in \cite{Va93copy}.}
Write $P{\Downarrow}$ if $P\Longmapsto P'$ for a process $P'$ with $P'{\downarrow}$.

\begin{definition}[\cite{Gorla10a}]\label{df:encoding}
An \emph{encoding} of $\mathcal{L}_1=(\mathcal{P}_1,\longmapsto_1,\asymp_1)$
into $\mathcal{L}_2=(\mathcal{P}_2,\longmapsto_2,\asymp_2)$ is a pair
$(\trans{\cdot},\varphi_{\subtrans})$ where
$\trans{\cdot}:\mathcal{P}_1\rightarrow\mathcal{P}_2$ is called \emph{translation}
and $\varphi_{\subtrans}:\N\rightarrow\N^k$ for some $k\mathbin\in\mbox{\bbb N}$
is called \emph{renaming policy} and is such that for $u\neq v$ the
$k$-tuples $\varphi_{\subtrans}(u)$ and $\varphi_{\subtrans}(v)$ have no name in common.
\end{definition}
\noindent
The terms of the source and target languages $\mathcal{L}_1$ and
$\mathcal{L}_2$ are often called $S$ and $\pT$, respectively.

\begin{definition}[\cite{Gorla10a}]\label{df:valid}
An encoding is \emph{valid} if it satisfies the following five criteria.
\begin{enumerate}
\item \emph{Compositionality:} for every $k$-ary operator $\op$ of
  $\mathcal{L}_1$ and for every set of names $N\subseteq\N$,
  there exists a $k$-ary context $C_\op^N[\__1;\dots;\__k]$ such that
  $$\trans{\op(S_1,\ldots,S_k)} = C_\op^N(\trans{S_1};\ldots;\trans{S_k})$$
  for all $S_1,\ldots,S_k\in\mathcal{P}_1$ with $\fn(S_1,\dots,S_n)=N$.
\item \emph{Name invariance:} for every $S\mathbin\in\mathcal{P}_1$ and
  $\sigma:\N\rightarrow\N$
  \[\begin{array}{cccl}
    \trans{S\sigma} & = & \trans{S}\sigma' & \mbox{if $\sigma$ is injective}\\
    \trans{S\sigma} & \asymp_2 & \trans{S}\sigma' & \mbox{otherwise}\\
    \end{array}\]
  with $\sigma'$ such that
  $\varphi_{\subtrans}(\sigma(a))\mathbin=\sigma'(\varphi_{\subtrans}(a))$
  for all $a\mathbin\in\N\!$.
\item \emph{Operational correspondence:}\\
  \begin{tabular}{@{}ll}
   \emph{Completeness} & if $S \Longmapsto_1 S'$ then $\trans{S}\Longmapsto_2\asymp_2 \trans{S'}$\\
   \emph{Soundness} & and if $\trans{S} \Longmapsto_2 \pT$ then $\exists S'\!:$\\&
   $S\Longmapsto_1 S'$ and $\pT\Longmapsto_2\asymp_2 \trans{S'}$.
  \end{tabular}
\item \emph{Divergence reflection:} 
  if $\trans{S} \longmapsto_2^\omega$ then $S \longmapsto_1^\omega$.
\item \emph{Success sensitiveness:}
  $S {\Downarrow}$ iff $\trans{S}{\Downarrow}$.\\
  For this purpose $\trans{\cdot}$ is extended to deal with the added constant $\surd$ by taking $\trans{\surd}=\surd$.
\end{enumerate}
\end{definition}

\section[Encoding synchronous into asynchronous pi]
        {Encoding synchronous into asynchronous $\pi$}\label{sec:pi2api}

\noindent
Consider the $\pi$-calculus as presented by Milner in \cite{Mi92}, i.e., the one of
Sangiorgi and Walker \cite{SW01book} without matching, $\tau$-prefixing, and choice.

Given a set of \emph{names} $\N$, the set $\mathcal{P}_\pi$ of \emph{processes} or \emph{terms} $P$
of the calculus is given by  $$P ::= \textbf{0}  ~~\mid~~ \bar xy.P ~~\mid~~ x(z).P ~~\mid~~ P|P' ~~\mid~~ (\nu z)P ~~\mid~~ !P$$
with $x,y,z,u,v,w$ ranging over $\N$.

\begin{definition}\label{def:structural congruence}
An occurrence of a name $z$ in $\pi$-calculus process $P\in\T_\pi$ is \emph{bound} if it lays
within a subexpression $x(z).P'$ or $(\nu z)P'$ of $P$; otherwise it is \emph{free}.
Let $\n(P)$ be the set of names occurring in $P\in \T_\pi$,
and $\fn(P)$ (resp.\ $\bn(P)$) be the set of names occurring free (resp.\ bound) in $P$.

\emph{Structural congruence}, $\equred$, is the smallest congruence relation on processes satisfying
\[
\begin{array}[b]{@{}l@{~}r@{~\!\equred\!~}l@{\hspace{1pt}}r@{~\!\equred\!~}l@{\,}r@{}}
\scriptstyle (1)& P_1 | (P_2 | P_3) & (P_1 | P_2) | P_3 &
(\nu z) \textbf{0} & \textbf{0} & \scriptstyle (5)\\

\scriptstyle (2)& P_1 | P_2 & P_2 | P_1 &
(\nu z)(\nu u)P & (\nu u)(\nu z) P & \scriptstyle (6)\\

\scriptstyle (3)& P | \textbf{0} & P &
(\nu w) (P | Q) & P | (\nu w)Q  & \scriptstyle (7)\\

& \multicolumn{2}{c}{} &
(\nu z)P & (\nu w)P\subs{w}{z} & \scriptstyle (8)\\

\scriptstyle (4)& !P & P | !P &
x(z).P & x(w).P\subs{w}{z}.  & \scriptstyle (9)\\
\end{array}
\]
Here $w\notin\n(P)$, and $P\subs{w}{z}$ denotes the process
obtained by replacing each free occurrence of $z$ in $P$ by $w$.

\end{definition}

\begin{definition}\label{def:reduction}
The \emph{reduction relation}, ${\longmapsto}\subseteq \T_\pi \times \T_\pi$, is generated by the
following rules.
\[\begin{array}{@{}cc@{}}
\displaystyle\frac{y\notin\bn(Q)}{\bar xy.P | x(z).Q \longmapsto P|Q\subs{y}{z}} &
\displaystyle\frac{P \longmapsto P'}{P|Q \longmapsto P'|Q} \\[3ex]
\displaystyle\frac{P \longmapsto P'}{(\nu z)P \longmapsto (\nu z)P'} &
\displaystyle\frac{Q \equred P \quad\! P \longmapsto P' \quad\! P' \equred Q'}{Q \longmapsto Q'}
\end{array}\]
\end{definition}

\noindent
The asynchronous $\pi$-calculus, as introduced by Honda \& Tokoro in \cite{HT91} and by Boudol in
\cite{Bo92}, is the sublanguage $\rm a\pi$ of the fragment $\pi$ of the $\pi$-calculus presented above where all
subexpressions $\bar x y.P$ have the form $\bar x y.\textbf{0}$, and
are written $\bar x y$. 

Boudol \cite{Bo92} defined an encoding $\trans{\cdot}_{\rm B}$ from $\pi$ to $\rm a\pi$ inductively as follows:
\[
\begin{array}{rcl}
\trans{\textbf{0}}_{\rm B}   &=& \textbf{0} \\
\trans{\bar{x}y.P}_{\rm B}   &=& (u)(\bar{x}u | u(v).(\bar{v}y | \trans{P}_{\rm B})) \\
\trans{x(z).P}_{\rm B}       &=& x(u).(v)(\bar{u}v | v(z).\trans{P}_{\rm B}) \\
\trans{P | Q}_{\rm B}        &=& (\trans{P}_{\rm B} | \trans{Q}_{\rm B}) \\
\trans{!P}_{\rm B}           &=& \ ! \trans{P}_{\rm B} \\
\trans{(\nu x) P}_{\rm B}    &=& (\nu x) \trans{P}_{\rm B}
\end{array}
\]
always choosing $u,v \notin\fn(P)\cup\{x,y\},~ u\neq v$.
The encoding $\trans{\cdot}_{\rm HT}$ of Honda \& Tokoro \cite{HT91} differs only in the
clauses for the input and output prefix:
\[
\begin{array}{rcl}
\trans{\bar{x}y.P}_{\rm HT}   &=& x(u).(\bar{u}y | \trans{P}_{\rm HT}) \\
\trans{x(z).P}_{\rm HT}      &=& (u)(\bar{x}u | u(z).\trans{P}_{\rm HT})
\end{array}
\]
again choosing $u\notin\fn(P)\cup\{x,y\}$.

\section{Validity of Boudol's encoding}\label{sec:valid}

\newcommand{\rn}{z}
\newcommand{\un}{y}

\noindent
In this section I show that Boudol's encoding satisfies all five
criteria of Gorla \cite{Gorla10a}. I will drop the subscript $_{\rm B}$.

\subsection{Compositionality}

\noindent
Boudol's encoding is compositional by construction, for it is
\emph{defined} in terms of the contexts $C_\op^N$ that are required to
exist by Definition~\ref{df:valid}. Note that, for the cases of input
and output prefixing, these contexts \emph{do} depend on $N$, namely
through the requirement that the fresh names $u$ and $v$ are chosen to lay outside $N$.

\subsection{Name invariance}

\noindent
An encoding according to Gorla is a pair $(\trans{\cdot},\varphi_{\subtrans})$,
of which the second component, the renaming policy, is relevant only
for satisfying the criterion of name invariance. Here I take
$k\mathbin=1$ and $\varphi_{\subtrans}:\N\rightarrow\N$ the identity mapping.

\begin{lemma}\label{free vars}\label{sbst}
Let $S\mathbin\in\T_{\pims}$. Then $\fn(\trans{S})=\fn(S)$.\\
Moreover, $\trans{S}\subs{y}{z} = \trans{S\subs{y}{z}}$ for any $y,z\in\mathcal N$.
\end{lemma}
\begin{proof}
A straightforward structural induction on $S$.
\end{proof}

\noindent
This implies that $\trans{S\sigma} \mathbin= \trans{S}\sigma$ for any renaming
$\sigma\!:\N\mathbin\rightarrow\N\!$, injective or otherwise.
So the criterion of name invariance is satisfied.

\subsection{Operational correspondence}

\noindent
A process calculus \`a la Gorla is a triple
$\mathcal{L}\mathbin=(\mathcal{P},\longmapsto,\asymp)$;
so far I defined $\mathcal{P}$ and $\longmapsto$ only.
The semantic equivalence $\asymp$ of the source language plays no
r\^ole in assessing whether an encoding is valid; the one of the
target language is used only for satisfying the criteria of name
invariance and operational correspondence. Here I take $\asymp_\pims$ and
$\asymp_\pima$ the identity relations.

If $S \equred S'$ for $S,S'\in\T_\pims$ then there exists a sequence
$S_0 \equred S_1 \equred \dots \equred S_n$ for some $n\geq 0$, with
$S=S_0$ and $S'=S_n$, such that each each step $S_i \equred S_{i+1}$
for $0\leq i<n$ is an application of one of the rules
$\scriptstyle(1)-(9)$ of Definition~\ref{def:reduction} or their
symmetric counterparts
$\scriptstyle\stackrel{\leftarrow}{(1)}-\stackrel{\leftarrow}{(9)}$.
(In fact, there is no need for rules
$\scriptstyle\stackrel{\leftarrow}{(2)},\stackrel{\leftarrow}{(6)},\stackrel{\leftarrow}{(8)}$
and $\scriptstyle\stackrel{\leftarrow}{(9)}$
as rules $\scriptstyle(2),(6),(8)$ and $\scriptstyle(9)$ are their own symmetric counterparts.)
Being an application of a rule $L\equred R$
here means that $S_i=C[L]$ and $S_{i+1}=C[R]$ for some unary context $C[\__1]$.

\paragraph{Operational completeness}

\begin{lemma}\label{redBF}
If $S \equred S'$ for $S,S' \mathbin\in \T_{\pims}$ then $\trans{S} \equred \trans{S'}$.
\end{lemma}
\begin{proof}
Using the reflexivity, symmetry and transitivity of $\equred$ one may restrict
attention to the case that $S \equred S'$ is a single application of a
rule $\scriptstyle(1)-(9)$ of Definition~\ref{def:reduction}.
The proof proceeds by structural induction on the context $C[\__1]$.
The case that $C[\__1] = \__1$, the \emph{trivial context}, is
straightforward for each of the rules $\scriptstyle(1)-(9)$,
applying Lemma~\ref{sbst} in the cases of rules $\scriptstyle(8),(9)$.
The induction step is a straightforward consequence of the
compositionality of $\trans{\cdot}$.
\end{proof}

\begin{lemma}\label{operational completeness}
Let $S,S' \mathbin\in \T_{\pims}$.
If $S \longmapsto S'$ then $\trans{S} \Longmapsto \trans{S'}$.
\end{lemma}

\begin{trivlist}\item[\hspace{4pt}\textit{Proof.}]
By induction on the derivation of $S \longmapsto S'$.
\begin{itemize}
\item Let $S\mathbin=\bar xy.P | x(z).Q$, $y\mathbin{\notin}\bn(Q)$ and $S'\mathbin=P|Q\subs{y}{z}$.
Pick $u,v\notin\fn(P)\cup\fn(Q)$, with $u\mathbin{\neq} v$.\\
Write $P^*\mathbin{:=}\bar{v}y | \trans{P}$ and $Q^*\mathbin{:=}v(z).\trans{Q}$. Then
\[\begin{array}{@{}rcl@{}}\trans{S} &=& (u)(\bar{x}u | u(v).P^*) ~|~ x(u).(v)(\bar{u}v | Q^*)\\
&\longmapsto& (u)\big(u(v).P^* ~|~ (v)(\bar{u}v | Q^*)\big)\\
&\longmapsto& (v)(P^* ~|~ Q^*)\\
&\longmapsto& \trans{P} ~|~ (\trans{Q}\subs{y}{z})\\
&=& \trans{P} ~|~ \trans{Q\subs{y}{z}} \quad~~~\mbox{(using Lemma~\ref{sbst})}\\
&=& \trans{P ~|~ Q\subs{y}{z}} = \trans{S'}.
\end{array}\]
Here structural congruence is applied in omitting parallel components $\bm{0}$ and empty binders
$(u)$, $(v)$.
\item Let $S = (z) P$ and $S' = (z) P'$, with $P \longmapsto P'$.
    By the induction hypothesis, $\trans{P} \Longmapsto \trans{P'}$.
    Therefore, $\trans{S} \Longmapsto \trans{S'}$,
    as $\trans{S} = (z) \trans{P}$ and $\trans{S'} = (z) \trans{P'}$.
\item The case that $S\mathbin= P|Q$ and $S'\mathbin= P'|Q$ with $P \longmapsto P'$
    proceeds likewise.
\item Let $S \equred P$ and $P' \equred S'$ with $P \longmapsto P'$.
    By the induction hypothesis, $\trans{P} \Longmapsto \trans{P'}$.
    By Lemma~\ref{redBF}, $\trans{S} \equred \trans{P}$ and $\trans{P'} \equred \trans{S'}$.
    So $\trans{S} \Longmapsto \trans{S'}$.
\qed
\end{itemize}
\end{trivlist}
The above yields that $S \Longmapsto S'$ implies $\trans{S} \Longmapsto \trans{S'}$.
So the criterion of operational completeness is satisfied.

\begin{remark}\label{divergence preservation}
The above proof shows that $\Longmapsto$ in Lemma~\ref{operational completeness}
may be replaced by $\longmapsto\longmapsto\longmapsto$. As a direct
consequence $S \longmapsto^\omega$ implies $\trans{S} \longmapsto^\omega$
(\emph{divergence preservation}).
\end{remark}

\paragraph{Operational soundness}
The following result provides a normal form up to structural
congruence for reduction steps in the asynchronous
$\pi$-calculus.  Here a term is \emph{plain} if it is a parallel
composition $P_1|\dots|P_n$ of subterms $P_i$ of the form
$\bar xy.R$ or $x(z).R$ or $\surd$ or $\nil$ or $!R$.
Moreover,
$(\tilde w)P$ for $\tilde w \mathbin=\{w_1,\dots,w_n\}\mathbin\subseteq \N$
with $n\mathbin\in\mbox{\bbb N}$ denotes $(w_1)\cdots(w_n)P$ for some
arbitrary order of the $(w_i)$.
Without the statements that $\C$ is plain and
$\tilde w \subseteq \fn((\bar x \un | x(\rn) . R) | \C )$, this lemma
is a simplification, by restricting attention to the syntax of
$\pima$, of Lemma 1.2.20 in~\cite{SW01book}, established for the full
$\pi$-calculus.

\begin{lemma}\label{red1}
If $\pT \longmapsto \pT'$ with $\pT,\pT'\in\T_{\pima}$ then there are
$\tilde w \mathbin\subseteq\N$, $x,\un,\rn\mathbin\in\mathcal{N}$
and terms $R,\C\mathbin\in\T_{\pima}$ with $\C$ plain, such that
$\pT \equred (\tilde w)((\bar x \un | x(\rn) . R) | \C ) \mathbin{\longmapsto}
(\tilde w)(( \bm{0} | R\subs{\un}{\rn}) | \C ) \equred \pT'$
and $\tilde w \subseteq \fn((\bar x \un | x(\rn) . R) | \C )$.
\end{lemma}
\begin{proof}
The reduction $\pT \longmapsto \pT'$ is provable from the reduction rules
of Definition~\ref{def:reduction}. Since $\equred$ is a congruence,
applications of the last rule can always be commuted until they appear
at the end of such a proof. Hence there are terms $\pT^{\rm pre}$ and $\pT^{\rm post}$
such that $\pT \equred \pT^{\rm pre} \longmapsto \pT^{\rm post} \equred \pT'$,
and the reduction $\pT^{\rm pre} \longmapsto \pT^{\rm post}$ is generated
by the first three rules of Definition~\ref{def:reduction}.
Applying rules $\scriptstyle(8)$, $\scriptstyle(9)$, $\scriptstyle(2)$ and
$\scriptstyle\stackrel{\leftarrow}{(7)}$ of structural congruence,
the terms $\pT^{\rm pre}$ and $\pT^{\rm post}$ can be brought in the forms
$(\tilde w)P^{\rm pre}$ and $(\tilde w)P^{\rm post}$, with
$P^{\rm pre}$ and $P^{\rm post}$ plain, at the same time moving all
applications of the reduction rule for restriction $(z)P$
after all applications of the rule for parallel composition. Applying rules
\plat{$\scriptstyle\stackrel{\leftarrow}{(1)},\stackrel{\leftarrow}{(3)}$},
all applications of the reduction rule for parallel composition can be
merged into a single application. After this proof normalisation,
the reduction $\pT^{\rm pre} \longmapsto \pT^{\rm post}$ is generated by
one application of the first reduction rule of Definition~\ref{def:reduction},
followed by one application of the rule for $|$, followed by
applications of the rule for restriction. Now $\pT^{\rm pre}$ has the
form $(\tilde w)((\bar x \un | x(\rn) . R) | \C)$ and
$\pT^{\rm post}=(\tilde w)((\bm{0} | R\subs{\un}{\rn}) |\C)$ with $\C$ plain.

Rules $\scriptstyle\stackrel{\leftarrow}{(3)}$, $\scriptstyle{(7)}$,
$\scriptstyle{(5)}$ and $\scriptstyle{(3)}$ of structural congruence,
in combination with $\alpha$-conversion (rules $\scriptstyle{(8)}$ and $\scriptstyle{(9)}$),
allow all names $w$ with $w \notin\fn((\bar x \un | x(\rn) . R) | \C)$
to be dropped from $\tilde w$, while preserving the validity of
$\pT^{\rm pre} \longmapsto \pT^{\rm post}$.
\end{proof}

\noindent
Write $P \eqused Q$ if $P$ can be converted into $Q$ using
applications of rules $\scriptstyle(1)-(3),(5)-(9)$ only, in either direction,
possibly within a context,
and $P \Rrightarrow_{!} Q$ if this can be done with rule $\scriptstyle{(4)}$, from left to right.

\begin{lemma}\label{red2}
  Lemma~\ref{red1} can be strengthened by replacing \mbox{$\pT\equred (\tilde w)((\bar x \un | x(\rn) . R) |\C) $} by
  $\pT \Rrightarrow_{!}\eqused (\tilde w)((\bar x \un | x(\rn) . R) |\C)$.
\end{lemma}
\begin{proof}
Define rule $\scriptstyle\stackrel{\leftarrow}{(4)}$ to \emph{commute over} rule $\scriptstyle(1)$
if for each sequence $P\equred Q \equred R$ with $P\equred Q$ an application of
rule $\scriptstyle\stackrel{\leftarrow}{(4)}$ and $Q\equred R$ an application of rule
$\scriptstyle(1)$, there exists a term $Q'$ such that $P\equred Q'$ holds by (possibly multiple)
applications of rule $\scriptstyle(1)$ and $Q'\equred R$ by applications of rule
$\scriptstyle\stackrel{\leftarrow}{(4)}$.
As indicated in the table below, rule $\scriptstyle\stackrel{\leftarrow}{(4)}$ commutes over all
other rules of structural congruence, except for rule $\scriptstyle(4)$.
The proof of this is trivial: in each case the two rules act on disjoint part of the syntax tree of $Q$.
Moreover, rule $\scriptstyle\stackrel{\leftarrow}{(4)}$ commutes over rule $\scriptstyle(4)$ too,
except in the special case that the two applications annihilate each other precisely, meaning that $P=R$;
this situation is indicated by the $\star$.
\[
\begin{array}{@{}r@{\,}||@{\,}c@{\,}|@{\,}c@{\,}|@{\,}c@{\,}|@{\,}c@{\,}|@{\,}c@{\,}|@{\,}c@{\,}|@{\,}c@{\,}|@{\,}c@{\,}|@{\,}c@{\,}|@{\,}c@{\,}|@{\,}c@{\,}|@{\,}c@{\,}|@{\,}c@{\,}|@{\,}c@{\,}|@{}}
&
\scriptstyle(1) & \scriptstyle\stackrel{\leftarrow}{(1)} &
\scriptstyle(2) &
\scriptstyle(3) & \scriptstyle\stackrel{\leftarrow}{(3)} &
\scriptstyle(4) & \scriptstyle\stackrel{\leftarrow}{(4)} &
\scriptstyle(5) & \scriptstyle\stackrel{\leftarrow}{(5)} &
\scriptstyle(6) &
\scriptstyle(7) & \scriptstyle\stackrel{\leftarrow}{(7)} &
\scriptstyle(8) &
\scriptstyle(9) \\
\hline
\scriptstyle\stackrel{\leftarrow}{(4)} & \scriptstyle\surd  & \scriptstyle\surd  & \scriptstyle\surd & \scriptstyle\surd  & \scriptstyle\surd  & \star & \cdot & \scriptstyle\surd  & \scriptstyle\surd  & \scriptstyle\surd  & \scriptstyle\surd  & \scriptstyle\surd  & \scriptstyle\surd  & \scriptstyle\surd  \\
\end{array}
\]
As a consequence of this, in a sequence
$P_0 \mathbin{\equred} P_1 \mathbin{\equred} \dots \mathbin{\equred} P_n$, all applications of rule
$\scriptstyle\stackrel{\leftarrow}{(4)}$ can be moved to the right. Moreover, when
$P_{n-1} \equred P_n \mathbin{:=} (\tilde w)((\bar x \un | x(\rn) . R) | \C ) \longmapsto\linebreak[1]
(\tilde w)(( \bm{0} | R\subs{\un}{\rn}) | \C ) \equred \pT'\!$
and $P_{n-1} \mathbin{\equred} P_n$ is an application of rule
\plat{$\scriptstyle\stackrel{\leftarrow}{(4)}$}, 
then this application must take place within the term $R$ or $\C$, and thus can be postponed until
after the reduction step, so that
$P_{n-1} \mathbin= (\tilde w)((\bar x \un | x(\rn) . R') | \C' ) \mathbin{\longmapsto}
(\tilde w)(( \bm{0} | R'\subs{\un}{\rn}) | \C' ) \mathbin{\equred} \pT'\!$
with $\C'$ plain.
Thus, one may assume that in the sequence $\pT=P_0 \equred P_1 \equred \dots \equred P_n$
none of the steps is an application of rule $\scriptstyle\stackrel{\leftarrow}{(4)}$.

Since applications of rule
$\scriptstyle\stackrel{\leftarrow}{(4)}$ could be shifted to the right in this
sequence, all applications of rule $\scriptstyle(4)$ can be shifted to the left.
Hence $\pT \mathbin{\Rrightarrow_{!}\eqused} (\tilde w)((\bar x \un | x(\rn) . R) | \C )$.
\end{proof}

\begin{lemma}\label{redBB}
If $\trans{S} \Rrightarrow_{!} \pT_0$ for $S \mathbin\in \T_{\pims}$ and $\pT_0 \mathbin\in \T_{\pima}$
then there is an $S_0 \mathbin\in \T_{\pims}$ with $S \Rrightarrow_{!} S_0$ and $\trans{S_0}=\pT_0$.
\end{lemma}
\begin{proof}
Similar to the proof of Lemma~\ref{redBF}.
\end{proof}

\noindent
Note that a variant of Lemma~\ref{redBB} with
$\scriptstyle(2)$, $\scriptstyle(3)$, $\stackrel{\leftarrow}{\scriptstyle(3)}$,
or $\stackrel{\leftarrow}{\scriptstyle(7)}$
in the r\^ole of $\scriptstyle{(4)}$ would not be valid.

Up to $\eqused$ each term $P\in\T_\pims$ can be brought in the form
$(\tilde w)P$ with $P$ plain and $\tilde w \subseteq \fn(P)$.
Moreover, such a normal form has a degree of uniqueness:
\begin{observation}\label{sigma}
If $(\tilde w)P \mathbin{\eqused} (\tilde v)Q$ with $P,Q$ plain,
$\tilde w \mathbin{\subseteq} \fn(P)$ and $\tilde v \mathbin\subseteq \fn(Q)$,
then there is an injective renaming $\sigma{:}\N{\rightarrow}\N$
such that $\sigma(\tilde v)=\tilde w$ and $P \eqused Q\sigma$.
Thus, for each parallel component $P'$ of $P$ of the form
$\bar xy.R$ or $x(z).R$ or $\surd$ or $!R$ there is a parallel component
$Q'$ of $Q\sigma$ with $P' \eqused Q'$.
\end{observation}
\noindent
Below, $\equiv_{(8),(9)}$ denotes convertibility by applications of rules
$\scriptstyle(8)$ and $\scriptstyle(9)$ only, and similarly for other rules.

\renewcommand{\rn}{r}
\renewcommand{\un}{u}

\begin{lemma}\label{red3}
If $\trans{S}\eqused (\tilde w)(\C| \bar x \un | x(\rn) . R )$ with $S\mathbin\in\T_{\pims}$,
$\C$ plain and $\tilde w \subseteq \fn(\C| \bar x \un | x(\rn) . R )$,
then there are $\D,R_1,\linebreak[4]R_2\mathbin\in\T_{\pims}$, $\F\mathbin\in\T_\pima$,
$y,z,v_1,v_2\mathbin\in\mathcal{N}$ and {$\tilde s, \tilde t \mathbin\subseteq {\mathcal{N}}$} 
such that
$S \eqused (\tilde s)(\D \mid \bar x y . R_1 \mid x(z) . R_2)$,
$v_1\mathbin{\neq}y\mathbin{\neq}u$,
$\tilde w = \tilde s \uplus \tilde t \uplus \{u\}$,
$\C \eqused \F|\un(v_1).(\bar{v_1}y | \trans{R_1})$,
$u,v_1\mathbin{\notin}\fn(\trans{R_1})$,
$\trans{\D}\eqused(\tilde t)\F$,
$R \eqused (v_2)(\bar{\rn}v_2 | v_2(z).\trans{R_2})$,
$r\mathbin{\neq} v_2$ and
$r,v_2\mathbin{\notin}\fn(\trans{R_2}){\setminus}\{z\}$.
\end{lemma}
\begin{proof}
By applying rules $\scriptstyle(8)$, $\scriptstyle(9)$, $\scriptstyle(2)$ and
$\scriptstyle\stackrel{\leftarrow}{(7)}$ only,
$S$ can be brought into the form $S':=(\tilde p)(P_1 | \dots | P_n)$
for some $n\mathbin>0$, where each
$P_i$ is of the form $\bar sy.R$ or $s(z).R$ or $\surd$ or $\nil$ or $!R$.
By means of \plat{$\scriptstyle\stackrel{\leftarrow}{(3)}$}, $\scriptstyle{(7)}$,
$\scriptstyle{(5)}$ and $\scriptstyle{(3)}$ one can moreover assure
that $\tilde p \subseteq \fn(P_1 | \dots | P_n)$.
By the proof of Lemma~\ref{redBF} $\trans{S}\eqused \trans{S'}$.
Furthermore, $\trans{S'}=(\tilde p)(\trans{P_1}|\dots|\trans{P_n})$.

By applying rules $\scriptstyle(8)$, $\scriptstyle(9)$, $\scriptstyle(2)$ and
$\scriptstyle\stackrel{\leftarrow}{(7)}$ only, the term $\trans{P_1}|\dots|\trans{P_n}$
can be brought into the form $(\tilde q)P$ with $P$ plain; moreover,
the set $\tilde q$ can be chosen disjoint from $\tilde p$.

Each $q\in\tilde q$ is a renaming of the name $u$ in a term
$\trans{P_i} = \trans{\bar{x}y.Q}  = (u)(\bar{x}u | u(v).(\bar{v}y | \trans{Q}))$,
so that $q \in\fn(P)$.

So $(\tilde w)(\C| \bar x \un | x(\rn) . R) \eqused (\tilde p)(\tilde q)P$.
Let $\sigma$ be the renaming that exists by
Observation~\ref{sigma}, so that
$\sigma(\tilde p) \uplus \sigma(\tilde q) = \tilde w$.
Then $\C| \bar x \un | x(\rn) . R \eqused P\sigma$.
So $\bar x \un$ and $x(\rn) . R$ (up to $\eqused$) must be parallel components of $P\sigma$.

Let $\sigma'$ be the restriction of $\sigma$ to $\tilde p$ and take $\tilde s := \sigma'(\tilde p)$.
Let $S''\mathbin{:=}(\tilde s)(P_1\sigma' | \dots | P_n\sigma')$.
Then $S' \eqused S''$
and $\trans{S'}\eqused\trans{S''}=(\tilde s)(\trans{P_1\sigma'}|\dots|\trans{P_n\sigma'})=(\tilde s)(\trans{P_1}\sigma'|\dots|\trans{P_n}\sigma')$.
Since $\trans{P_1}|\dots|\trans{P_n}$ can be converted
into $(\tilde q)(P)$ by applications of rules
$\scriptstyle(8)$, $\scriptstyle(9)$, $\scriptstyle(2)$ and
\plat{$\scriptstyle\stackrel{\leftarrow}{(7)}$},
$\trans{P_1}\sigma'|\dots|\trans{P_n}\sigma'$ can be converted
into $(\tilde q)(P\sigma')$ and even into $(\sigma(\tilde q))(P\sigma)$
by applications of these rules. One can apply
$\scriptstyle(8),(9)$ first, so that each $\trans{P_i}\sigma'$ is
converted into some term $Q_i$ by applications of $\scriptstyle(8),(9)$,
and $Q_1|\dots|Q_n$ is converted into $(\sigma(\tilde q))(P\sigma)$ by
applications of $\scriptstyle(2)$ and
\plat{$\scriptstyle\stackrel{\leftarrow}{(7)}$} only.

One of the $P_i\sigma'$ must be of the form $\bar x y.R_1$,
so that $\trans{P_i}\sigma' = \trans{P_i\sigma'} = \trans{\bar x y.R_1}=(u')(\bar{x}u' | u'(v_1).(\bar{v}_1y | \trans{R_1}))$
with 
$u',v_1\notin\fn(R_1)\cup\{x,y\}$,
while $u'$ is renamed into $u$ in $Q_i = (u)(\bar{x}u | u(v_1).(\bar{v}_1y | \trans{R_1}))$.
For this is the only way $\bar{x}u$ can end up as a parallel component of $P\sigma$.
It follows that $u,v_1 \notin \fn(R_1)$ and
$v_1\mathbin{\neq} y\mathbin{\neq} u\in\sigma(\tilde q)$.
Let $\tilde t := \sigma(\tilde q) \setminus u$.
One obtains $\tilde w = \tilde s \uplus \tilde t \uplus \{u\}$.

Searching for an explanation of the parallel component $x(\rn) . R $ (up to $\eqused$)
of $P\sigma$,
the existence of the component $\bar x u$ of $P\sigma$
excludes the possibility that
one of the $P_i\sigma'$ is of the form $\bar t' y'\!.R_2$
so that $\trans{P_i}\sigma' {=} (u')(\bar{t}'u' | u'(r'\hspace{-.5pt}).(\bar{r'}y' | \trans{R_2}))$, while
$u'$ is renamed into $x$ and $r'$ into $r$ in the expression
$Q_i = (x)(\bar{t}'x | x(r).(\bar{r}y' | \trans{R_2}))$.

Hence one of the $P_i\sigma'$ is of the form $x(z).R_2$, so that
$\trans{P_i}\sigma' = \trans{P_i\sigma'} = \trans{x(z).R_2}=x(r').(v')(\bar{r'}v' | v'(z).\trans{R_2})$
with $r'\neq v'$ and $r',v'\notin\fn(R_2)\setminus\{z\}$, while 
$r'$, $v'$ and $z$ are renamed into $r$, $v_2$ and $z'$ in
$Q_i = x(r).(v_2)(\bar{r}v_2 | v_2(z').R_2')$, where
$(z')R_2' \mathbin{\equiv_{(8),(9)}} (z)\trans{R_2}$.
Thus $r,v_2\notin\fn(R_2){\setminus}\{z\}$ and\linebreak $r{\neq} v_2$.
Further, $x(z).R \mathbin{\eqused} Q_i$, so 
$R\mathbin{\eqused} (v_2)(\bar{r}v_2 | v_2(z).\trans{R_2})$.

Let $\D$ collect all parallel components $P_i\sigma'$ other than the
above discussed components $\bar x y.R_1$ and $x(z).R_2$. Then
$S\eqused S''\eqused (\tilde s)(\D ~|~ \bar x y.R_1 ~|~ x(z).R_2)$.

One has $(\tilde w)(\C| \bar x \un | x(\rn) . R ) \eqused \trans{S} \eqused \trans{S''} \eqused\\
\trans{(\tilde s)(\D ~|~ \bar x y.R_1 ~|~  x(z).R_2)} =\\
(\tilde s)(\trans{\D} ~|~ (u')(\bar{x}u' | u'(v_1).(\bar{v_1}y | \trans{R_1})) ~|~ x(r')(v')\cdots ) \equiv_{(8),(9)}
(\tilde s)(\E ~|~ (u)(\bar{x}u | u(v_1).(\bar{v_1}y | \trans{R_1})) ~|~ x(\rn) . R ) \equiv_{(1),(2),(7)} \\
(\tilde s)(u)(\E ~|~ u(v_1).(\bar{v_1}y | \trans{R_1}) ~|~ \bar{x}u ~|~ x(\rn) . R ) \equiv_{(2),(7)} \\
(\tilde s)(u)(\tilde t)(\F ~|~ u(v_1).(\bar{v_1}y | \trans{R_1}) ~|~ \bar{x}u ~|~ x(\rn) . R )$.\\
Here $\E$ is the parallel composition of all components $Q_i$ obtained by renaming of the parallel
components  $P_i\sigma'$ of $\trans{\D}$, and $(\tilde t)\F$ with $\F$ plain is obtained from $\E$ by rules
$\scriptstyle(2),(7)$.
So $\C ~|~ \bar x \un ~|~ x(\rn) . R \eqused \F ~|~ u(v_1).(\bar{v_1}y | \trans{R_1}) ~|~ \bar{x}u ~|~ x(\rn) . R$
by Observation~\ref{sigma}.
It follows that $\C \eqused \F ~|~ u(v_1).(\bar{v_1}y | \trans{R_1})$.
\end{proof}

\noindent
A straightforward case distinction shows that the set of names occurring free in a term is invariant
under structural congruence:
\begin{observation}\label{obs:subst}
If $P\equred Q$ then $\fn(P)=\fn(Q)$.
\end{observation}

\noindent
The above results can be combined to establish the special case of
operational soundness where the sequence of reductions
$\trans{S}\Longmapsto \pT$ consists of one reduction step only.

\begin{lemma}\label{lem6}
Let $S \in \T_{\pims}$ and $\pT \in \T_{\pima}$.
If $\trans{S} \longmapsto \pT$ then there is a $S'$ with $S \longmapsto S'$
and $\pT\Longmapsto \trans{S'}$.
\end{lemma}

\begin{proof}
Suppose $\trans{S} \longmapsto \pT$. Then, by Lemma~\ref{red2}, there are
$\tilde w \mathbin\subseteq\N$, $x,\un,\rn\mathbin\in\mathcal{N}$
and $\pT_0,R,\C\mathbin\in\T_{\pima}$ with $\C$ plain, such that
$\trans{S} \mathbin{\Rrightarrow_{!}} \pT_0 \mathbin{\eqused}
(\tilde w)(\C| \bar x \un | x(\rn) . R ) \mathbin{\longmapsto}
(\tilde w)(\C| R\subs{\un}{\rn} ) \mathbin{\equred} \pT$
and $\tilde w \subseteq \fn(\C| \bar x \un | x(\rn) . R)$.
By Lemma~\ref{redBB}, there is an $S_0 \mathbin\in \T_{\pims}$ with $S \Rrightarrow_{!} S_0$ and $\trans{S_0}=\pT_0$.
So, by Lemma~\ref{red3}, there are
$\D,R_1,\linebreak[1]R_2\mathbin\in\T_{\pims}$, $\F\mathbin\in\T_\pima$,
$y,z,v_1,v_2\mathbin\in\mathcal{N}$ and {$\tilde s, \tilde t \mathbin\subseteq {\mathcal{N}}$} 
such that
$S \eqused (\tilde s)(\D \mid \bar x y . R_1 \mid x(z) . R_2)$,
$v_1\mathbin{\neq}y\mathbin{\neq}u$,
$\tilde w = \tilde s \uplus \tilde t \uplus \{u\}$,
$\C \eqused \F|\un(v_1).(\bar{v_1}y | \trans{R_1})$,
$u,v_1\mathbin{\notin}\fn(\trans{R_1})$,
$\trans{\D}\eqused(\tilde t)\F$,
$R \eqused (v_2)(\bar{\rn}v_2 | v_2(z).\trans{R_2})$,
$r\mathbin{\neq} v_2$ and
$r,v_2\mathbin{\notin}\fn(\trans{R_2}){\setminus}\{z\}$.

As $(\fn(\D) \cup \{x,y\} \cup \fn(R_1)\cup(\fn(R_2)\setminus\{z\}))\setminus \tilde s = \fn(S_0)$,
using Observation~\ref{obs:subst},
and
$\tilde w \cap \fn(S_0) = \tilde w \cap \fn(\pT_0) \mathbin=\emptyset$, using Lemma~\ref{sbst},
$t,u\notin\fn(\D) \cup \{x,y\} \cup \fn(R_1)\cup(\fn(R_2)\setminus\{z\})$
for all $t\mathbin\in \tilde t$.
Let $v\mathbin\in\N$ satisfy $u,r,y\mathbin{\neq} v\notin\fn(R_1)\cup(\fn(R_2){\setminus}\{z\})$.

Take $S':= (\tilde s)(\D ~|~ R_1 ~|~ R_2 \subs{y}{z})$.
Then $S\longmapsto S'$ and
$\pT \mathbin{\equred} (\tilde w)(\C~|~R\subs{u}{r})\\
\equred (\tilde w)\big(\F | \un(v_1).(\bar{v}_1y | \trans{R_1}) ~|~
(v_2)(\bar{\rn}v_2 | v_2(z).\trans{R_2})\subs{u}{r}\big) \\
\equred (\tilde w)\big(\F | \un(v).(\bar{v}y | \trans{R_1}) ~|~
(v)(\bar{\rn}v | v(z).\trans{R_2})\subs{u}{r}\big) \\
\mbox{}\hfill (\mbox{\it as $y\mathbin{\neq} v_1$, $v_1\mathbin{\notin}\fn(\trans{R_1})$,
                     $r\mathbin{\neq} v_2$ and $v_2\mathbin{\notin}\fn(\trans{R_2}){\setminus} \{z\}$})\\
\equred (\tilde w)\big(\F | \un(v).(\bar{v}y | \trans{R_1}) ~|~
(v)(\bar{u}v | v(z).\trans{R_2})\big) \\
\mbox{}\hfill (\mbox{\it since $r\neq v \neq u$ and $r\notin\fn(\trans{R_2}){\setminus}\{z\}$})\\
\equred (\tilde s)(u)(\tilde t)\big(\F  ~|~ \un(v).(\bar{v}y | \trans{R_1}) ~|~
(v)(\bar{u}v | v(z).\trans{R_2})\big) \\
\equred (\tilde s)(u)\big((\tilde t)\F  ~|~ \un(v).(\bar{v}y | \trans{R_1}) ~|~
(v)(\bar{u}v | v(z).\trans{R_2})\big) \\
\mbox{}\hfill (\mbox{\it since  $t\mathbin{\notin}\{u,y\} \cup \fn(\trans{R_1})\cup(\fn(\trans{R_2})\setminus \{z\})$ for $t\mathbin\in\tilde t$})\\
\equred (\tilde s)\left(\trans{\D} ~|~ (u){\color{red}\big(}\un(v).(\bar{v}y | \trans{R_1}) ~|~
(v)(\bar{u}v | v(z).\trans{R_2}){\color{red}\big)}\right) \\
\mbox{}\hfill (\mbox{\it since $u\notin\fn(\trans{\D})$})\\
\longmapsto (\tilde s)\big(\trans{\D} ~|~ (v).{\color{red}\big(}(\bar{v}y | \trans{R_1}) ~|~
v(z).\trans{R_2}{\color{red}\big)}\big) \\
\mbox{}\hfill (\mbox{\it since $u \neq v$ and $u\notin\{y\}\cup\fn(\trans{R_1})\cup\fn(\trans{R_2}){\setminus}\{z\}$})\\
\longmapsto (\tilde s)\big(\trans{\D} ~|~ \trans{R_1} ~|~ \trans{R_2}\subs{y}{z} \big) \\
\mbox{}\hfill (\mbox{\it since $v\notin\{y\}\cup\fn(\trans{R_1})\cup\fn(\trans{R_2}){\setminus}\{z\}$})\\
= (\tilde s)\big(\trans{\D} ~|~ \trans{R_1} ~|~ \trans{R_2\subs{y}{z}} \big)
\hfill (\mbox{\it by Lemma~\ref{sbst}})\\
= \trans{S'}$.
\end{proof}

\noindent
To obtain general operational soundness, I introduce an \emph{inert}
reduction relation with a confluence property, stated in Lemma~\ref{lem2} below:
any other reduction that can occur as an alternative to an inert
reduction is still possible after the occurrence of the inert reduction.

\newcommand{\XX}{P}
\newcommand{\YY}{Q}
\begin{definition}\label{def:cool-rel}
Let $\equiv\!\Rrightarrow$ be the smallest relation on $\T_{\pima}$ such that
\begin{enumerate}
\vspace{-2pt}
\item $(v)(\bar vy | \XX | v(z).\YY) \equiv\!\Rrightarrow \XX|(\YY\subs{y}{z})$,
\vspace{-2pt}
\item if $\XX \equiv\!\Rrightarrow \YY$ then $\XX|R \equiv\!\Rrightarrow \YY|R$,
\vspace{-2pt}
\item if $\XX \equiv\!\Rrightarrow \YY$ then $(w) \XX \equiv\!\Rrightarrow (w) \YY$,
\vspace{-2pt}
\item if $\XX \equred \XX' \equiv\!\Rrightarrow \YY' \equred \YY$ then $\XX \equiv\!\Rrightarrow \YY$,
\end{enumerate}
where
$v \not\in \fn(\XX) \cup \fn(\YY\subs{y}{z})$.
\end{definition}

\noindent
First of all observe that whenever two processes are related by $\equiv\!\Rrightarrow$, an actual
reduction takes place.

\begin{observation}\label{lem1}
If $P \equiv\!\Rrightarrow Q$ then $P \longmapsto Q$.
\end{observation}

\noindent
As its proof shows, the conclusion of Lemma~\ref{operational completeness}
can be restated as $\trans{S} \longmapsto\equiv\!\Rrightarrow\equiv\!\Rrightarrow \trans{S'}$.
Likewise, the two occurrences of $\longmapsto$ at the end of the proof of
Lemma~\ref{lem6} can be replaced by $\equiv\!\Rrightarrow$:

\begin{observation}\label{inert}
In the conclusion of Lemma~\ref{lem6}, $\pT\Longmapsto \trans{S'}$ can
be restated as $\pT \equiv\!\Rrightarrow\equiv\!\Rrightarrow \trans{S'}$.
\end{observation}

\begin{lemma}\label{lem2}
If $P \equiv\!\Rrightarrow Q$ and $P \longmapsto P'$ with $P'\not\equred Q$
then there is a $Q'$ with $Q \longmapsto Q'$ and $P' \equiv\!\Rrightarrow Q'$.
\end{lemma}

\begin{proof}
By Lemma~\ref{red1} there are
$\tilde w \subseteq\N$, $x,y,z\in\mathcal{N}$
and $R,\C\in\T_{\pima}$ such that $\C$ plain, $\tilde w \subseteq \fn((\bar x y | x(z) . R) | \C )$ and
$P \mathbin{\equred} P_0 \mathbin{:=} (\tilde w)((\bar x y | x(z) . R) | \C ) \mathbin{\longmapsto}
(\tilde w)(( \bm{0} | R\subs{y}{z}) | \C ) \mathbin{\equred} P'\!$.
Likewise, there are
$\tilde q \subseteq\N$, $x',y',z'\in\mathcal{N}$
and $R',\C'\in\T_{\pima}$ such that $\C'$ plain, $\tilde q \subseteq \fn((\bar x' y' | x'(z') . R') | \C' )$ and
$P \mathbin{\equred} P_1 \mathbin{:=} (\tilde q)((\bar x' y' | x'(z') . R') | \C' ) \mathbin{\longmapsto}
(\tilde q)(( \bm{0} | R'\subs{y'}{z'}) | \C' ) \equred Q$.
So $(\tilde w)(( \bm{0} | R\subs{y}{z}) | \C ) \longmapsfrom P_0 \equred P_1 \longmapsto
(\tilde q)(R'\subs{y'}{z'} | \C' )$.
As in the proof of Lemma~\ref{red2}, all applications of rule \plat{$\scriptstyle\stackrel{\leftarrow}{(4)}$}
in the sequence of reductions $P_0 \equred P_1$ can be moved to the right and shifted over
the $\longmapsto$. Likewise,  all applications of rule \plat{$\scriptstyle{(4)}$}
can be moved to the left and shifted over the $\longmapsfrom$. Therefore, I may assume
that $P_0 \eqused P_1$.
Let $\sigma$ be the injective renaming that exists by Observation~\ref{sigma}. Then
$(\bar x y | x(z) . R) | \C \eqused ((\bar x' y' | x'(z') . R') | \C')\sigma$.
Let $u:=\sigma(x')$, $v:=\sigma(y')$, $r:=\sigma(z')$, $R'':=R\sigma$ and $\C'':=\C\sigma$.
Then $(\bar x y | x(z) . R) | \C \eqused (\bar u v | u(r) . R'') | \C''$.

In a reduction step of the form $P \equiv\!\Rrightarrow Q$, the reacting prefixes $\bar a b$ and
$a(c).\D$ are always found in the scope of a restriction operator $(a)$, and without a $!$ between $(a)$
and $\bar a b$ or $a(c).\D$, such that in this scope
there are no other unguarded occurrences of prefixes $\bar a d$ or $a(e).\F$.
This follows by a trivial induction on the definition of $\equiv\!\Rrightarrow$.
In particular, this property is preserved when applying structural congruence to $P$.
Consequently, the plain term $\C''$ has no parallel components of the form
$\bar u y''$ or $u(z'').R'''$.

Case 1: $x=u$. Then, by the above, $(z)R \eqused (r)R''$, $y=v$ and $\C \eqused \C''$.
Consequently, $P' \equred Q$.

Case 2: $x\neq u$. Then $\bar u v$ and $u(r) . R''$ (up to $\eqused$) must be parallel components
of $\C$, and $\bar x y$ and $x(z) . R$ (up to $\eqused$) must be parallel components
of $\C''$, so that
$P \equred P_0 = (\tilde w)((\bar x y | x(z) . R) | (\bar u v | u(r) . R'') | \D)$,
where $\C\eqused  (\bar u v | u(r) . R'') | \D$ and $\C'' \eqused (\bar x y | x(z) . R) | \D$.
This shows that the reductions $P \equiv\!\Rrightarrow Q$ and $P \longmapsto P'$ are concurrent,
so that there is a $Q'$ with $Q \longmapsto Q'$ and $P' \equiv\!\Rrightarrow Q'$.
\end{proof}

\begin{corollary}\label{lem2 transitive}
If $P \equiv\!\Rrightarrow Q$ and $P \Longmapsto P'$
then either $Q\Longmapsto P'$ or there is a $Q'$
with $Q \Longmapsto Q'$ and $P' \equiv\!\Rrightarrow Q'$.
Moreover, the sequence $Q\Longmapsto P'$ or $Q\Longmapsto Q'$
contains no more reduction steps than the sequence $P \Longmapsto Q'$.
\end{corollary}
\begin{proof}
By repeated application of Lemma~\ref{lem2}.
\weg{
By induction on the length $n$ of the sequence $P \Longmapsto P'$.
The base case $n=0$ is trivial: take $Q':=Q$.
So let $P \longmapsto P_1 \Longmapsto P'$, where $P_1 \Longmapsto P'$  has length $n$.
By Lemma~\ref{lem2} either $P_1 \equred Q$ or
there is a $Q_1$ with $Q \longmapsto Q_1$\linebreak[4] and $P_1 \equiv\!\Rrightarrow Q_1$.
In the first case $Q \Longmapsto P'$, and the length of this sequence is $n<n{+}1$.
In the second case, by induction, either $Q_1\Longmapsto P'$ or there is a $Q'$
with $Q_1 \Longmapsto Q'$ and $P' \equiv\!\Rrightarrow Q'$.
Moreover, the sequence $Q_1\Longmapsto P'$ or $Q_1\Longmapsto Q'$ has length $\leq n$.
Thus, either $Q\Longmapsto P'$ or there is a $Q'$
with $Q \Longmapsto Q'$ and $P' \equiv\!\Rrightarrow Q'$.
Moreover, the sequence $Q\Longmapsto P'$ or $Q\Longmapsto Q'$ has length $\leq n{+}1$.
}
\end{proof}
\begin{corollary}\label{cor lem2 transitive}
If $P \equiv\!\Rrightarrow^* Q$ and $P \Longmapsto P'$
then there is a $Q'$ with $Q \Longmapsto Q'$ and $P' \equiv\!\Rrightarrow^* Q'$.
Moreover, the sequence $Q\Longmapsto Q'$
contains no more reduction steps than the sequence $P \Longmapsto Q'$.
\end{corollary}

\noindent
By combining Corollary~\ref{cor lem2 transitive} with Observations~\ref{lem1}
and~\ref{inert} one finds that the criterion of operational soundness is met.

\begin{theorem}\label{operational soundness}
Let $S \in \T_{\pims}$ and $\pT \in \T_{\pima}$.
If $\trans{S} \Longmapsto \pT$ then there is a $S'$ with $S \Longmapsto S'$
and $\pT\Longmapsto \trans{S'}$.
\end{theorem}

\begin{proof}
By induction on the length $n$ of the sequence\linebreak $\trans{S} \Longmapsto \pT$.
The base case $n\mathbin=0$ is trivial: take $S':=S$.
So let $\trans{S} \longmapsto \pT_1 \Longmapsto \pT$, where $\pT_1 \Longmapsto \pT$ has length $n$.
By Lemma~\ref{lem6} with Observation~\ref{inert} $\exists S_1$ with $S \longmapsto S_1$
and $\pT_1\equiv\!\Rrightarrow^* \trans{S_1}$.
By Corollary~\ref{cor lem2 transitive} $\exists \pT'$ with
$\trans{S_1} \Longmapsto \pT'$ and $\pT \equiv\!\Rrightarrow^* \pT'$.
Furthermore, the sequence $\trans{S_1} \Longmapsto \pT'$ has length $\leq n$.
By induction, there is a $S'$ with $S_1 \mathbin{\Longmapsto} S'$ and $\pT'\Longmapsto \trans{S'}$.
Hence $S \Longmapsto S'$ and $\pT\Longmapsto \trans{S'}$, using Observation~\ref{lem1}.
\end{proof}

\subsection{Divergence reflection}

\begin{corollary}\label{divergence}
If $P \equiv\!\Rrightarrow Q$ and $P \longmapsto^\omega$
then $Q \longmapsto^\omega$.
\end{corollary}
\begin{proof}
By repeated application of Lemma~\ref{lem2}.
\end{proof}

\noindent
Together with Observation~\ref{inert} this implies that the criterion of divergence reflection is met.

\begin{theorem}
Let $S \in \T_{\pims}$.
If $\trans{S} \longmapsto^\omega$ then $S \longmapsto^\omega$.
\end{theorem}

\begin{proof}
Suppose $\trans{S} \longmapsto^\omega$. Then $\trans{S} \longmapsto \pT_1 \longmapsto^\omega$.
By Lemma~\ref{lem6} with Observation~\ref{inert} there is an $S_1$ with $S \mathbin{\longmapsto} S_1$
and $\pT_1\mathbin{\equiv\!\Rrightarrow^*} \trans{S_1}$.
By Corollary~\ref{divergence} $\trans{S_1} \longmapsto^\omega$.
In the same way there is an $S_2$ with $S_1 \mathbin{\longmapsto} S_2$
and $\trans{S_2} \longmapsto^\omega$, and so on. Thus $S \longmapsto^\omega$.
\end{proof}

\subsection{Success sensitiveness}\label{success}

\noindent
The success predicate $\downarrow$ can also be defined inductively:
\[
\sbarb{\surd}{\bar x} \qquad
\frac{\sbarb{P}{\bar x}}{\sbarb{(P|Q)}{\bar x}} \qquad
\frac{\sbarb{Q}{\bar x}}{\sbarb{(P|Q)}{\bar x}} \qquad
\frac{\sbarb{P}{\bar x}}{\sbarb{((\nu z)P)}{\bar x}} \qquad
\frac{\sbarb{P}{\bar x}}{\sbarb{(!P)}{\bar x}}
\]
Note that if $P \equred Q$ and $P{\downarrow}$ then also $Q{\downarrow}$.

\begin{lemma}\label{barbs}
Let $S \in \T_{\pims}$. Then $\trans{S}{\downarrow}$ iff $S{\downarrow}$.
\end{lemma}
\begin{proof}
A trivial structural induction.
\end{proof}

\begin{lemma}\label{barb confluence}
If $\pT \longmapsto \pT'$ and $\pT{\downarrow}$ then $\pT'{\downarrow}$.
\end{lemma}

\begin{proof}
By Lemma~\ref{red1} there are
$\tilde w \subseteq\N$, $x,\un,\rn\mathbin\in\mathcal{N}$
and $R,\C\mathbin\in\T_{\pima}$ with $\C$ plain, such that
$\tilde w \subseteq \fn((\bar x \un | x(\rn) . R) | \C )$ and
$\pT \mathbin{\equred} (\tilde w)((\bar x \un | x(\rn) . R) | \C ) \mathbin{\longmapsto}
(\tilde w)(( \bm{0} | R\subs{\un}{\rn}) | \C ) \mathbin{\equred} \pT'$.
Since $\pT{\downarrow}$, it must be that $\C{\downarrow}$ and hence $\pT'{\downarrow}$.
\end{proof}

\noindent
By combining Lemmata~\ref{barbs} and~\ref{barb confluence} with Lemma~\ref{operational completeness}
and Theorem~\ref{operational soundness} one finds that also the criterion of success
sensitiveness is met.

\begin{theorem}
Let $S \in \T_{\pims}$.
Then $S{\Downarrow}$ iff $\trans{S}{\Downarrow}$.
\end{theorem}

\begin{proof}
Suppose that $S{\Downarrow}$. Then $S \Longmapsto S'$ for a process $S'$ with $S'{\downarrow}$.
By Lemma~\ref{operational completeness} $\trans{S} \Longmapsto \trans{S'}$.
By Lemma~\ref{barbs} $\trans{S'}{\downarrow}$.
Hence $\trans{S}{\Downarrow}$.

Now suppose $\trans{S}{\Downarrow}$. 
Then $\trans{S} \Longmapsto \pT$ for a process $\pT$ with $\pT{\downarrow}$.
By Theorem~\ref{operational soundness} there is a $S'$ with $S \Longmapsto S'$
and $\pT\Longmapsto \trans{S'}$.
By Lemma~\ref{barb confluence} $\trans{S'}{\downarrow}$.
By Lemma~\ref{barbs} $S'{\downarrow}$.
Hence $S{\Downarrow}$.
\end{proof}

\section{Validity of Honda \& Tokoro's encoding}\label{sec:validHT}

\noindent
That the encoding of Honda \& Tokoro also satisfies all five
criteria of Gorla follows in the same way. I will only show the steps where a
difference with the previous sections occurs. In this section $\trans{\cdot}$ stands for $\trans{\cdot}_{\rm HT}$.

\begin{lemma}\label{operational completeness HT}
Let $S,S' \mathbin\in \T_{\pims}$.
If $S \longmapsto S'$ then $\trans{S} \Longmapsto \trans{S'}$.
\end{lemma}

\begin{trivlist}\item[\hspace{4pt}\textit{Proof.}]
By induction on the derivation of $S \longmapsto S'$.
\begin{itemize}
\item Let $S\mathbin=\bar xy.P | x(z).Q$, $y\mathbin{\notin}\bn(Q)$ and $S'\mathbin=P|Q\subs{y}{z}$.
Pick $u\notin\fn(P)\cup\fn(Q)\cup\{x,y\}$. Then
\[\begin{array}{@{}rcl@{}}\trans{S} &=& x(u).(\bar{u}y | \trans{P}) ~|~ (u)(\bar{x}u | u(z).\trans{Q})\\
&\longmapsto& (u)\big(\bar{u}y | \trans{P} ~|~ u(z).\trans{Q}\big)\\
&\longmapsto& \trans{P} ~|~ (\trans{Q}\subs{y}{z})\\
&=& \trans{P} ~|~ \trans{Q\subs{y}{z}} \quad~~~\mbox{(using Lemma~\ref{sbst})}\\
&=& \trans{P ~|~ Q\subs{y}{z}} = \trans{S'}.
\end{array}\]
Here structural congruence is applied in omitting parallel components $\bm{0}$ and the empty binders $(u)$.
\item The other three cases proceed as in the proof of Lemma~\ref{operational completeness}.
\qed\pagebreak[2]
\end{itemize}
\end{trivlist}

\begin{lemma}\label{red3HT}
If $\trans{S}\eqused (\tilde w)(\C| \bar x \un | x(\rn) . R )$ with $S\mathbin\in\T_{\pims}$,
$\C$ plain and $\tilde w \subseteq \fn (\C| \bar x \un | x(\rn) . R )$,
then there are terms $\D,R_1,R_2\mathbin\in\T_{\pims}$, $\F\mathbin\in\T_\pima$, and names
$y,z\mathbin\in\mathcal{N}$ and {$\tilde s, \tilde t \subseteq {\mathcal{N}}$} 
such that
$S \eqused (\tilde s)(\D \mid \bar x y . R_1 \mid x(z) . R_2)$,
$\tilde w = \tilde s \uplus \tilde t \uplus \{u\}$,\linebreak[4]
$\C \eqused \F|\un(z).\trans{R_2}$,
$u\mathbin{\notin}\fn(\trans{R_2}){\setminus}\{z\}$,
$\trans{\D}\eqused(\tilde t)\F$,
$r \neq y$,
$R \eqused \bar{r}y | \trans{R_1}$ and
$r\mathbin{\notin}\fn(\trans{R_1})$.
\end{lemma}
\begin{proof}
The first two paragraphs proceed exactly as in the proof of Lemma~\ref{red3}.

Each $q\in\tilde q$ is a renaming of the name $u$ in a term
$\trans{P_i} = \trans{x(z).Q}  = (u)(\bar{x}u | u(z).\trans{Q})$,
so that $q \in\fn(P)$.

The next two paragraphs proceed exactly as in the proof of Lemma~\ref{red3}.

One of the $P_i\sigma'$ must be of the form $x (z).R_2$,
so that $\trans{P_i}\sigma' = \trans{P_i\sigma'} = \trans{x (z).R_2}=(u')(\bar{x}u' | u'(z).\trans{R_2})$
with $u'\notin\fn(R_2){\setminus}\{z\}\cup\{x\}$,
while $u'$ and $z$ are renamed into $u$ and $z'$ in $Q_i = (u)(\bar{x}u | u(z').R'_2)$,
where $(z')R_2' \mathbin{\equiv_{(8),(9)}} (z)\trans{R_2}$.
For this is the only way $\bar{x}u$ can end up as a parallel component of $P\sigma$.
It follows that $u \notin \fn(R_2){\setminus}\{z\}$ and $u\in\sigma(\tilde q)$.
Let $\tilde t := \sigma(\tilde q) \setminus u$.
One obtains $\tilde w = \tilde s \uplus \tilde t \uplus \{u\}$.

Searching for an explanation of the parallel component $x(\rn) . R $ (up to $\eqused$)
of $P\sigma$,
the existence of the component $\bar x u$ of $P\sigma$
excludes the possibility that
one of the $P_i\sigma'$ is of the form $t'(r')\!.R_1$
so that $\trans{P_i}\sigma' {=} (u')(\bar{t}'u' | u'(r').\trans{R_1})$, while
$u'$ is renamed into $x$ and $r'$ into $r$ in the expression
$Q_i = (u')(\bar{t}'x | x(r).R'_1)$.

Hence one of the $P_i\sigma'$ is of the form $\bar xy.R_1$, so that
$\trans{P_i}\sigma' \mathbin= \trans{P_i\sigma'} \mathbin= \trans{\bar xy.R_1}\mathbin=x(r').(\bar r'y|\trans{R_1})$
with $r'\mathbin{\notin}\fn(R_1)\cup\{x,y\}$, while 
$r'$ is renamed into $r$ in
$Q_i = x(r).(\bar ry|\trans{R_1})$.
Thus $r\notin\fn(R_1)$ and $r\neq y$.
Further, 
$R\mathbin{\eqused} (\bar ry|\trans{R_1})$.

Let $\D$ collect all parallel components $P_i\sigma'$ other than the
above discussed components $\bar x y.R_1$ and $x(z).R_2$. Then
$S\eqused S''\eqused (\tilde s)(\D ~|~ \bar x y.R_1 ~|~ x(z).R_2)$.

One has $(\tilde w)(\C| \bar x \un | x(\rn) . R ) \eqused \trans{S} \eqused \trans{S''} \eqused\\
\trans{(\tilde s)(\D ~|~ \bar x y.R_1 ~|~  x(z).R_2)} =\\
(\tilde s)(\trans{\D} ~|~ x(r').(\bar r'y|\trans{R_1}) ~|~ (u')(\bar{x}u' | u'(z).\trans{R_2}) ) \equiv_{(8),(9)}
(\tilde s)(\E ~|~ x(\rn) . R ~|~ (u)(\bar{x}u | u(z).\trans{R_2}) ) \equiv_{(1),(2),(7)} \\
(\tilde s)(u)(\E ~|~ \bar{x}u | u(z).\trans{R_2} ~|~ x(\rn) . R ) \equiv_{(2),(7)} \\
(\tilde s)(u)(\tilde t)(\F ~|~  u(z).\trans{R_2} ~|~ \bar{x}u ~|~ x(\rn) . R )$.\\
Here $\E$ is the parallel composition of all components $Q_i$ obtained by renaming of the parallel
components  $P_i\sigma'$ of $\trans{\D}$, and $(\tilde t)\F$ with $\F$ plain is obtained from $\E$ by rules
$\scriptstyle(2),(7)$.
So $\C ~|~ \bar x \un ~|~ x(\rn) . R \eqused \F ~|~ u(z).\trans{R_2} ~|~ \bar{x}u ~|~ x(\rn) . R$
by Observation~\ref{sigma}.
It follows that $\C \eqused \F ~|~ u(z).\trans{R_2}$.
\end{proof}

\begin{lemma}\label{lem6HT}
Let $S \in \T_{\pims}$ and $\pT \in \T_{\pima}$.
If $\trans{S} \longmapsto \pT$ then there is a $S'$ with $S \longmapsto S'$
and $\pT\Longmapsto \trans{S'}$.
\end{lemma}

\begin{proof}
The first paragraph proceeds as in the proof of Lemma~\ref{lem6}, but
incorporating the conclusion of Lemma~\ref{red3HT} instead of Lemma~\ref{red3}.
Again, one finds, for all $t\mathbin\in \tilde t$, that
$t,u\notin\fn(\D) \cup \{x,y\} \cup \fn(R_1)\cup(\fn(R_2)\setminus\{z\})$.

Take $S':= (\tilde s)(\D ~|~ R_1 ~|~ R_2 \subs{y}{z})$.
Then $S\longmapsto S'$ and
$\pT \mathbin{\equred} (\tilde w)(\C~|~R\subs{u}{r})\\
\equred (\tilde w)\big(\F|\un(z).\trans{R_2} ~|~
\bar{r}y | \trans{R_1}\subs{u}{r}\big) \\
\equred (\tilde w)\big(\F|\un(z).\trans{R_2} ~|~
\bar{u}y | \trans{R_1}\big) \\
\mbox{}\hfill (\mbox{\it since $r\mathbin{\neq} y $ and $r\notin\fn(\trans{R_1})$})\\
\equred (\tilde s)(u)(\tilde t)\big(\F~|~\un(z).\trans{R_2} ~|~
\bar{u}y | \trans{R_1}\big) \\
\equred (\tilde s)(u)\big((\tilde t)\F~|~ \bar{u}y | \trans{R_1} ~|~ \un(z).\trans{R_2} \big) \\
\mbox{}\hfill (\mbox{\it since  $t\mathbin{\notin}\{u,y\} \cup \fn(\trans{R_1})\cup(\fn(\trans{R_2})\setminus \{z\})$ for $t\mathbin\in\tilde t$})\\
\equred (\tilde s)\big(\trans{\D}~|~(u){\color{red}\big(}\bar{u}y | \trans{R_1} ~|~ \un(z).\trans{R_2} {\color{red}\big)}\big)
\mbox{}\hfill (\mbox{\it as $u\notin\fn(\trans{\D})$})\\
\longmapsto (\tilde s)\big(\trans{\D} ~|~ \trans{R_1} ~|~ \trans{R_2}\subs{y}{z} \big) \\
\mbox{}\hfill (\mbox{\it since $u\notin\{y\}\cup\fn(\trans{R_1})\cup\fn(\trans{R_2}){\setminus}\{z\}$})\\
= (\tilde s)\big(\trans{\D} ~|~ \trans{R_1} ~|~ \trans{R_2\subs{y}{z}} \big)
\hfill (\mbox{\it by Lemma~\ref{sbst}})\\
= \trans{S'}$.
\end{proof}

\noindent
As its proof shows, the conclusion of Lemma~\ref{operational completeness HT}
can be restated as $\trans{S} \longmapsto\equiv\!\Rrightarrow \trans{S'}$.
The occurrence of $\longmapsto$ at the end of the proof of
Lemma~\ref{lem6HT} can be replaced likewise:

\begin{observation}\label{inertHT}
In the conclusion of Lemma~\ref{lem6}, $\pT\Longmapsto \trans{S'}$ can
be restated as $\pT \equiv\!\Rrightarrow \trans{S'}$.
\end{observation}

\section{Conclusion}\label{sec:conclusion}

\noindent
This paper proved the validity according to Gorla of the encodings proposed by Boudol and by Honda \& Tokoro of
the synchronous choice-free $\pi$-calculus into its asynchronous fragment; that is, both encodings
enjoy the five correctness criteria of \cite{Gorla10a}. For such a result, easily believed to be ``obvious'',
the proofs are surprisingly complicated,\footnote{The complications lay chiefly with proving
  operational soundness; for some of the other criteria the proofs are straightforward.}
and involve the new concept of an inert reduction.
Yet, I conjecture that it is not possible to simplify the proofs in any meaningful way.

Below I reflect on three of Gorla's criteria in the light of the lessons learned from this case study.

\paragraph{Compositionality}
Compositionality demands that for every $k$-ary operator $\op$ of the source language
there is a $k$-ary context $C_\op^N[\__1;\dots;\__k]$ in the target such that
  $$\trans{\op(S_1,\ldots,S_k)} = C_\op^N(\trans{S_1};\ldots;\trans{S_k})$$
for all $S_1,\ldots,S_k\in\mathcal{P}_1$. A drawback of this criterion is that this context may
depend on the set of names $N$ that occur free in the arguments $S_1,\ldots,S_k$.
The present application shows that we cannot simply strengthen the criterion of compositionality by
dropping the dependence on $N$. For then the present encodings would fail to be compositional.
However, in \cite{vG12} a form of compositionality is proposed where $C_\op$ does not depend on $N$,
but the main requirement is weakened to
  $$\trans{\op(S_1,\ldots,S_k)} \stackrel\alpha= C_\op(\trans{S_1};\ldots;\trans{S_k}).$$
Here $\stackrel\alpha=$ denotes equivalence up to \emph{$\alpha$-conversion}, renaming of bound
names and variables, here corresponding with rules $\scriptstyle{(8)}$ and $\scriptstyle{(9)}$ of
structural congruence. This suffices to rescue the current encodings. It is an open question whether
there are examples of intuitively valid encodings that essentially need the dependence of $N$ allowed
by \cite{Gorla10a}, i.e., where $C_\op^{N_1}$ and $C_\op^{N_2}$ differ by more than $\alpha$-conversion.

Another method of dealing with the fresh names $u$ and $v$ that are used in the present encodings,
also proposed in \cite{vG12}, is to equip the target language with two fresh names that do not occur
in the set of names available for the source language. Making the dependence on the choice of set
$\N$ of names explicit, this method calls $\pims$ expressible into $\pima$ if for each $\N$ there
exists an $\N'$ such that there is a valid encoding of $\pi(\N)$ into $\pima(\N')$.

\paragraph{Operational soundness}

\hspace{-.8pt}Operational soundness stems from Nestmann \& Pierce \cite{NestmannP00}, who proposed two forms of it:\\
$(\mathfrak{I})$\hfill
if $\trans{S} \longmapsto_2 \pT$ then $\exists S'\!:$ $S\longmapsto_1 S'$ and $\pT\asymp_2 \trans{S'}$.~~\,\mbox{}\\
$(\mathfrak{S})$\hfill
if $\trans{S} \Longmapsto_2 \pT$ then $\exists S'\!:$ $S\Longmapsto_1 S'$ and $\pT\Longmapsto_2\trans{S'}$.\\
The version of Gorla is the common weakening of these:\\
$(\mathfrak{G})$\hfill
if $\trans{S} \mathbin{\Longmapsto_2} \pT$ then $\exists S'\!:\!S\mathbin{\Longmapsto_1} S'$ and $\pT\Longmapsto_2\asymp_2\trans{S'}$.\\
An interesting intermediate form is\\
$(\mathfrak{W})$\hfill
if $\trans{S} \Longmapsto_2 \pT$ then $\exists S'\!:$ $S\Longmapsto_1 S'$ and $\pT\asymp_2 \trans{S'}$.~~\,\mbox{}\\
Nestmann \& Pierce observed that ``nonprompt encodings'', that
``allow administrative (or book-keeping) steps to precede a committing step'',
``do not satisfy $(\mathfrak{I})$''. For such encodings, which include the ones studied here,
they proposed $(\mathfrak{S})$.
As I have shown, the encodings of Boudol and of Honda \& Tokoro, satisfy not only $(\mathfrak{G})$
but even $(\mathfrak{S})$. It remains an interesting open question whether they satisfy $(\mathfrak{W})$.
Clearly, they do not when taking $\asymp_2$ to be the identity relation---as I did here---or
structural congruence. However, it is conceivable that $(\mathfrak{W})$ holds for another reasonable
choice of $\asymp_2$. (An unreasonable choice, such as the universal relation, tells us nothing.)

\paragraph{Success sensitivity}

My treatment of success sensitivity differs slightly from the one of Gorla \cite{Gorla10a}.
Gorla requires $\surd$ to be a constant of any two languages whose expressiveness is
compared. Strictly speaking, this does not allow his framework to be applied to the encodings of
Boudol or Honda \& Tokoro, as these deal with languages not featuring $\surd$.  Here I simply allowed
$\surd$ to be added, which is in line with the way Gorla's framework has been used
\cite{Gorla10b,LPSS10,PSN11,PN12,PNG13,GW14,EPTCS160.4,EPTCS189.9,GWL16}.
A consequence of this decision is that I have to specify how $\surd$ is translated---see the last
sentence of Definition~\ref{df:valid}---as the addition of $\surd$ to both languages happens after a
translation is proposed. This differs from \cite{Gorla10a}, where it is explicitly allowed to take
$\trans{\surd} \neq \surd$.

Gorla's success predicate is one of the possible ways to provide source and target languages with
a set of \emph{barbs} $\Omega$, each being a unary predicate on processes. For $\omega \in \Omega$,
write $P{\downarrow_\omega}$ if process $P$ has the barb $\omega$, and 
$P{\Downarrow_\omega}$ if $P\Longmapsto P'$ for a process $P'$ with $P'{\downarrow_\omega}$.
The standard criterion of barb sensitivity is then
$S{\Downarrow_\omega} \Leftrightarrow \trans{S}{\Downarrow_\omega}$ for all $\omega\in\Omega$.

A traditional choice of barb in the $\pi$-calculus is to take $\Omega=\{x,\bar x \mid x\mathbin\in\N\}$,
writing $P{\downarrow_x}$, resp.\ $P{\downarrow_{\bar x}}$, when $x\mathbin\in\fn(P)$ and
$P$ has an unguarded occurrence of a subterm $x(z).R$, resp.\ $\bar xy.R$ \cite{SW01book}.
The philosophy behind the asynchronous $\pi$-calculus entails that input actions $x(z)$ are not
directly observable (while output actions can be observed by means of a matching input of
the observer). This leads to semantic identifications like $\nil = x(y).\bar xy$, for in
both cases the environment may observe $\bar x z$ only if it supplied $\bar x z$ itself first.
Yet, these processes differ on their input barbs ($\downarrow_x$). For this reason, in
$\rm a\pi$ normally only output barbs ${\downarrow_{\bar x}}$ are considered \cite{SW01book}.
Boudol's encoding satisfies the
criterion of output barb sensitivity (and in fact also input barb sensitivity).
However, the encoding of Honda \& Tokoro does not, as it swaps input and output barbs.
As such, it is an excellent example of the benefit of the external barb $\surd$.

\subsection*{Validity up to a semantic equivalence}

\noindent
In \cite{vG12} a compositional encoding is called \emph{valid up to} a semantic equivalence
${\sim} \subseteq \T \times \T$ if $\trans{P} \sim P$ for all $P\in\T$.\footnote{\cite{vG12}
  distinguishes between a process term and its meaning or denotation, $\sim$ being
  defined on the denotations. Here, in line with \cite{Gorla10a} and most other related work,
  I collapse syntax and semantics by only considering process terms without process variables,
  taking the denotation of a term to be itself (up to $\alpha$-conversion).}
A given encoding may be valid up to a coarse equivalence, and invalid up to a finer one.
The equivalence for which it is valid is then a measure of the quality of the encoding.

Combining the results of the current paper with those of \cite{EPTCS190.4} shows that the
encodings of Boudol and Honda \& Tokoro are valid up to
reduction-based \emph{success respecting coupled similarity} ($\textit{CS}^\surd$).
Earlier, \cite{CC01} established that Boudol's encoding is valid up to \emph{may testing}
\cite{DH84} and \emph{fair testing equivalence} \cite{BRV95,NC95}---both results now follow from the
validity up to \plat{$\textit{CS}^\surd$}.
On the other hand, \cite{CC01} also shows that Boudol's encoding is not valid up to a
form of \emph{must testing}; in \cite{CCP07} this result is strengthened to pertain to any encoding
of $\pims$ into $\pima$.

An interesting open question is whether the encodings of Boudol and Honda \& Tokoro are valid up to
reduction-based \emph{success respecting weak bisimilarity} or \emph{weak barbed bisimilarity}.
In \cite{QW00}, a polyadic version of Boudol's encoding is assumed to be valid up to the version of
weak barbed bisimilarity that uses output barbs only; see Lemma 17. Yet, as no proof is provided,
the question remains open.

\paragraph{Acknowledgement}
Fruitful feedback from Stephan Mennicke and from IPL referees is gratefully acknowledged.
\vspace{-3pt}

\bibliographystyle{eptcs}
\bibliography{../../../../Stanford/lib/abbreviations,../../../../Stanford/lib/dbase,ipl}
\end{document}